\documentclass[11pt]{article}
\usepackage[margin=1in]{geometry} 

\usepackage{amsthm}
\usepackage{amsmath}
\usepackage{booktabs} % For pretty tables
\usepackage{caption} % For caption spacing
\usepackage{subcaption} % For sub-figures
\usepackage{graphicx}
\usepackage{adjustbox}
\usepackage{pgfplots}
\usepackage[all]{nowidow}
\usepackage[utf8]{inputenc}
\usepackage{tikz}
\usetikzlibrary{er,positioning,bayesnet}
\usepackage{multicol}
\usepackage{algpseudocode,algorithm,algorithmicx}
\usepackage{booktabs}
\usepackage{float}
\usepackage{algorithm}
\floatname{algorithm}{Protocol}
\usepackage{authblk}

\newtheorem{theorem}{Theorem}[section]
\newtheorem{definition}[theorem]{Definition}

\usepackage[inline]{enumitem} % Horizontal lists
\definecolor{blue}{HTML}{1F77B4}
\definecolor{orange}{HTML}{FF7F0E}
\definecolor{green}{HTML}{2CA02C}

\pgfplotsset{compat=1.14}

\begin{document}
\title{A note on blind contact tracing at scale \\with applications to the COVID-19 pandemic}

\author[1]{Jack K. Fitzsimons}
\author[2]{Atul Mantri}
\author[3]{Robert Pisarczyk}
\author[4]{Tom Rainforth}
\author[5]{Zhikuan Zhao}

\affil[1]{Oblivious AI, Dublin, Ireland}
\affil[2]{Joint Center for Quantum Information and Computer Science (QuICS), University of Maryland, College Park, USA}
\affil[3]{Mathematical Institute, University of Oxford, Oxford, UK}

\affil[4]{Department of Statistics, University of Oxford, Oxford, UK}
\affil[5]{Department of Computer Science, ETH Zürich, Zürich, Switzerland}

\maketitle              % typeset the header of the contribution
\begin{abstract}
The current COVID-19 pandemic highlights the utility of contact tracing, when combined with case isolation and social distancing, as an important tool for mitigating the spread of a disease \cite{ferretti2020quantifying}. Contact tracing provides a mechanism of identifying individuals with a high likelihood of previous exposure to a contagious disease, allowing additional precautions to be put in place to prevent continued transmission. 

Here we consider a cryptographic approach to contact tracing based on secure two-party computation (2PC). We begin by considering the problem of comparing a set of location histories held by two parties to determine whether they have come within some threshold distance while at the same time maintaining the privacy of the location histories. We propose a solution to this problem using pre-shared keys, adapted from an equality testing protocol due to Ishai et al \cite{ishai2013power}. We discuss how this protocol can be used to maintain privacy within practical contact tracing scenarios, including both app-based approaches and approaches which leverage location history held by telecoms and internet service providers. We examine the efficiency of this approach and show that existing infrastructure is sufficient to support anonymised contact tracing at a national level. 
\end{abstract}

\noindent \textbf{Important note: } The authors of this manuscript come from backgrounds in computer science and physics and claim no special expertise in public health or epidemiology. The purpose of this paper is to explore the application of multi-party cryptography to the problem of contact tracing. As such, all references to the role of contact tracing in public health and to the COVID-19 pandemic more generally are purely to provide context, and should not be taken as authoritative.

\section{Introduction}
Coronavirus disease 2019 (COVID-19) has spread rapidly around the globe, with many governments and health organisations battling to slow or contain the spread of the disease \cite{ferguson2020report, walker2020report}. In particular, geofencing and contact tracing have been used extensively to try to contain the virus and to identify those potentially at risk due to direct contact with confirmed cases. The aim of contact tracing is two-fold: to help and manage people who may have been exposed to infectious disease, and to stop the chain of transmission in order to control the outbreak.

The prevalence of mobile telephones, and in particular smartphones, provide a potentially powerful tool for contact tracing efforts: the vast majority of users carry their devices with them throughout the day, and the phones themselves are capable of generating detailed location histories, either via GPS or signal strength .

However, any possible implementation using them has to tackle several issues. 
Firstly, there is an issue of privacy. With the availability of users' location, there is a growing concern that private data can be used by an unauthorized entity in ways that infringe the individual right to privacy, for example by profiling user's behavior, targeting the user by stalking and spamming or drawing inferences about user’s sensitive data. Unfortunately, there is an ongoing battle between personal privacy, the approaches used by governments, and the effectiveness of the approaches applied. For example, in South Korea and Israel, there has been a large public push back against the sharing of individuals' intimate information and the use of mass surveillance approaches respectively \cite{bbc}. 

Secondly, any scheme that is used has to be both scalable and effective. For example, schemes that rely on a smartphone application that has to be installed by users on their smartphones can only be effective if a significant proportion of the population has a smartphone, takes an active step to install the app, and allows it to run continuously. Further, in many areas, a minority of people have access to smartphones but may carry other GSM mobile devices that can be a source of information. For example in 2018 India had a smartphone penetration rate of 26\% in comparison to 88\% subscription rate to wireless phones \cite{mckinsey2019report}. Even in the countries which do not have this problem, there can be important groups in the society, such as older people, who have smaller smartphone ownership rates; this is critically important for COVID-19 as older people are in a high risk group \cite{berenguer2016smartphones,who_statement}. Such solutions can also struggle to get a large enough proportion of smartphone owners to use the app. Furthermore, if they rely on a single source of data such as Bluetooth, they might miss many other ones that are more widespread such as GPS or GSM triangulation.

\subsection{Summary of our results}

In this work we offer a practical framework for identifying whether two parties have come in close direct contact based on their absolute locations without disclosing either parties' location history.
To achieve this, we have developed a simple framework which utilizes techniques from secure two-party computation (2PC), whereby one party holds a database of private location data within a virtual private cloud or firewall and the other party holds a database of locations of confirmed COVID-19 patient locations for the same time period. The first database might be coming, for example, from a telecommunications service provider while the second is held by the government (see Figure~\ref{fig:overview_fig}). The goal is to identify if any of the locations in the former database lie within a threshold distance of locations in the latter database without revealing either database directly. In the literature this problem is known under the name of private proximity testing \cite{narayanan2011location}

The approach taken in this work is to first reduce the problem of private proximity testing into a blind contact tracing by an appropriate tessellation of the surface of the Earth \cite{narayanan2011location, sinnott1984sky}. We then blindly compute a number of matches between each row of the two data sets over a fixed period of time (e.g. two weeks). In this way, we also hide the exact timestamp of the matching. We achieve this by modifying the perfectly secure protocol for equality testing presented in \cite{ishai2013power}.

We offer three use cases of how the protocol may be used in practical application to enable collaboration between government and individuals or IT/telecom companies. Finally, we outline the implementation details for such an approach, showing that the resource intensity of such protocols is easily within capacity of modern data centers. 

%Our contribution is motivated by a number of factors. Firstly, the use of hashing of locations can be problematic as the set size of locations is small and thus creating reverse lookup tables is readily available to malicious actors. Secondly, many public key encryption protocols and approaches such as that which we propose requires effective implementations of random number generation albeit random integers or random ``good" primes. Finally, key reuse may offer apparent efficiency speed ups but this opens the protocol to potential frequency based attacks on user data. As we believe that even a single poorly implemented protocol that compromises the publics privacy could severely impact public trust in the proposed government-private party collaboration, we propose using a semi-trusted key generator to distribute keys to both parties for one time use. We endeavour to propose an efficient approach to this which can be understood by the non-expert audience.

\begin{figure*}[t!]
    \centering
    \includegraphics[width=\textwidth]{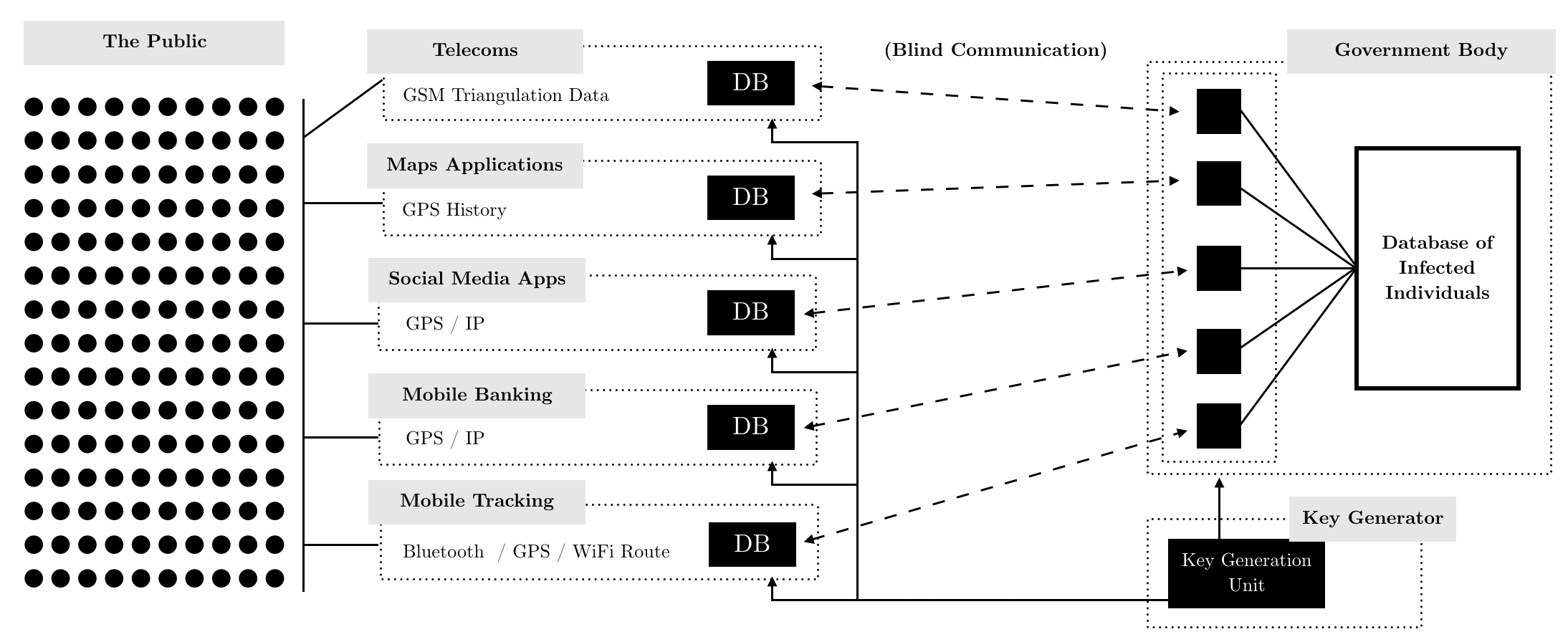}
    \caption{Outline of some of the parties who perform data collection of their users location with the proposed blind contact tracing interface between privately held databases and government. \vspace{5pt}}
    \label{fig:overview_fig}
\end{figure*}

\subsection{Related work}

Many institutions and governments have already introduced contact tracing smartphone applications or are currently working on them. Singapore has released TraceTogether~\cite{TraceTogether}, Poland published code for a prototype of their app~\cite{ProteGo}, British NHS and Germany's Fraunhofer Institute for Telecommunications are in the completion phases of their solution, all of which are based on the Bluetooth technology \cite{bell2020tracesecure}. Chan \textit{et al}.~\cite{chan2020pact} present a thorough review of contact tracing smartphone apps based on Bluetooth and presents a new third-party free protocol. The authors specifically do not consider approaches based on absolute-location, such as ours, due to a cryptographic difficulty of protecting infected patients' locations. However, we believe our approach offers a simple and efficient solution that achieves that objective. The solution closely related to ours has been proposed by MIT~\cite{MIT}. Their smartphone---Private Kit: Safe Paths---builds location trail from GPS data of users. In its current iteration, data of infected patients is shared with the government who redacts them. This is a similar setup to ours however we stress that our solution offers perfect security against active adversaries. In its third iteration, Private Kit: Safe Paths will not rely on sharing information with the government but will instead perform a multi-party computation based on location trails of users. We do not know the proposed protocol, and therefore do not comment on its security, however both iterations require high participation rates for the solution to be effective. 

Upon completion of this work, we became aware of the independent work by Berke \textit{et al} \cite{berke2020assessing}, which presents an approach for contact tracing problem via the private-set intersection problem. As far as we know, this is the only result that achieves the same functionality (c.f. Definition~\ref{def:bct} below) in the context of contact tracing. In the setting of computational security, the authors introduce a hash-based protocol with only two parties: the sender (server) and the receiver (client) without requiring the trusted dealer. They also provide an adaption of a Diffie-Hellman protocol as an example. On the other hand, we take the information-theoretic approach with the added assumption of trusted randomness dealer.  

Other cryptographic articles that are mostly related to ours concern private proximity testing. In particular, Narayanan \textit{et al}.~\cite{narayanan2011location} did a reduction of private proximity testing protocols to private equality testing. However it outputs time and position to the receiver, which we try to avoid in our setting. Šeděnka and Gasti~\cite{vsedvenka2014privacy} have studied proximity testing on Earth based on distance computation. However their functionality necessarily has a security weakness that allows for triangulation.

\section{Technical description the protocol}

\subsection{Setup}
In this article, we are interested in the secure computation for two-parties with a trusted dealer. The task of a trusted dealer is to instantiate a randomness distribution phase that will allow the two parties (sender and receiver) to securely compute a joint function. 
To aid with exposition, we will use a running example with a receiver called Alice and a sender called Bob.

Given the receiver and sender locations, denoted as $\{p_{a_t}\}$ and $\{p_{b_t}\}$ respectively,
%of two parties, a receiver and a sender respectively, 
the receiver learns how many times she was in proximity to the sender i.e.~how many times $D(p_{a_t},p_{b_t})<\Delta$, where $D$ is a distance measure, $\Delta$ is the threshold, and we assume that the locations are provided at discrete time points.
The sender receives no output and additionally, we require both the parties  not to learn anything about the other party's location or database. This is a modified blind proximity testing problem. However, we can reduce it to a blind contact tracing problem by tessellating the surface of the Earth (the choice of the grid for which is described in Section~\ref{grid}). After doing this, the locations of Alice and Bob $\{x_{t}\},\{y_{t}\}$ are discrete and correspond to the unit cells in the grid.  This reduction allows us to consider the joint function as equality on the input of the sender and receiver. The goal is to check for all $t$ if $x_t = y_t$ and output the total number of matches, $N$,  to the receiver without leaking anything about the $t$. However, when the receiver's input matches the sender's input for all the $t$ then receiver can trivially guess the sender's location. This functionality is inspired by the privacy-preserving equality testing where output $y_t$ is leaked to the receiver when $x_t = y_t$. Note that both the functionalities of equality testing and contact tracing achieve the same task  of privacy-preserving equality with the same security against the malicious sender but different security requirements when receiver is adversarial. The ideal functionality for blind contact tracing can be defined in the following way.

\begin{definition}{\textbf{Ideal functionality for blind contact tracing}}.
\label{def:bct}
The ideal functionality for blind contact tracing $\mathcal{S}_{bct}(\{x_{t}\},\{y_{y}\})$ for a receiver with input $\{x_t\}$ and sender with input $\{y_t\}$ is given by the following procedure:
%two parties Alice (receiver) with input $\{x_{t}\}$ and Bob (sender) with input $\{y_{t}\}$ can be defined in the following way:
\begin{enumerate}
    \item Take in receiver's input $\{x_{t}\}$ and sender's input $\{y_{t}\}$
    \item Compute $N = \sum_{t}s_t$ where $s_t = 1$ if $x_{t}=y_{t}$ else $s_i = 0$
    \item Return $N$ to the receiver and nothing to the sender.
\end{enumerate}
This provides correctness and perfect security against dishonest parties (sender or receiver). 
% . $\mathcal{S}_{bct}$ takes receiver's input $\{x_{t}\}$ and sender's input $\{y_{t}\}$. It computes  $M = \sum_{t}s_t$ where $s_t = 1$ if $x_{t}=y_{t}$ else $s_i = 0$. The ideal functionality outputs $M$ on receiver's side. The sender does not get any output. The ideal functionality $\mathcal{S}_{bct}$ provides correctness and perfect security against dishonest parties (sender or receiver). 
\end{definition}

\subsection{Presentation of the protocol}
Our concrete protocol for the blind contact tracing is adapted from the equality testing protocol of Ishai et al.~\cite{ishai2013power}. The blind contact tracing protocol is divided into two phases: key generation and computation. The first stage consists of three parties - trusted dealer, sender, receiver and proceeds in the following way. The sender and receiver requests correlated randomness from the trusted dealer and sends a pre-agreed (between sender and receiver) input $t$. The trusted dealer samples a tuple of permutations $(P_t, Q)$ where $P_t$ is 2-wise independent permutation chosen from a family of permutation and $Q$ is uniformly independent, chosen from $\mathcal{S}_n$. It also generates uniformly random string $r_t$. The trusted dealer sends the set $(\{P_t\}, Q)$ to sender and $(\{r_t\}, \{s_{Q(t)}\})$ to receiver where $s_{Q(t)} = P_t(r_{t})$. 

Upon receiving the randomness in the key generation phase, the sender and receiver proceed with the computation phase. It consists of one round of communication between the sender and the receiver. As a first step: receiver performs a one-time pad of their private input $\{x_t\}$ with the random strings $\{r_t\}$ to obtain $u_t$ where $u_t =  x_t +r_t$. The receiver sends $\{u_t\}$ to the sender. In the next step, the sender performs one-time pad of their private input $\{y_t\}$ with the string obtained from the receiver in the first step ($u_t$) and permutes it using the permutation $P_t$. The sender sends $v_{Q(t)}$ where $v_{Q(t)} = P_t (u_t-y_t)$ to the receiver. As the final step, the receiver checks if $s_{Q(t)}$ is equivalent to the string $v_{Q(t)}$ received from sender for every $t$. The receiver calculates the total number of matches $N$. We describe the blind contact tracing in Protocol \ref{main protocol}. 

\vspace{\baselineskip}

\begin{algorithm}
\caption{Blind Contact Tracing Protocol}
\vspace{\baselineskip}
\textbf{Receiver's input:} private location data $\{x_t \in X | t \in [n]\}$\\
\textbf{Sender's input:} private location data $\{y_t \in X | t \in [n]\}$\\
\textbf{Receiver's output:} N (total number of matches)\\
\begin{enumerate}
\item{\textbf{Key generation (trusted dealer, sender, receiver):}}
\begin{enumerate}
    \item Sample random 2-wise independent permutations ${P_t: X \mapsto X}$ and uniformly independent permutation $Q \in \mathcal{S}_n$ where $\mathcal{S}_n$ is the $n$th symmetric group and random strings ${r_t \in X}$.
    \item The sender gets $(\{P_t\}, Q)$ and the receiver gets $(\{r_t\}, \{s_{Q(t)}\})$, where $s_{Q(t)} = P_t(r_{t})$.
\end{enumerate}
\item {\textbf{Computation (sender, receiver):}}
\begin{enumerate}
\item The receiver computes $u_t = x_t +r_t$ and sends the collection $\{u_t\}$ to sender.
\item The sender computes $v_{Q(t)} = P_t (u_t-y_t)$ and sends the permuted collection $\{v_{Q(t)}\}$ to the receiver.
\item The receiver then compares $s_{Q(t)}$ to $v_{Q(t)}$ and outputs $N$, where $N:=\sum_{Q(t)} w_{Q(t)}$ and $w_{Q(t)} = 1$ if $s_{Q(t)} = v_{Q(t)}$ otherwise $w_{Q(t)} = 0$.

\end{enumerate}
\end{enumerate}
\label{main protocol}
\end{algorithm}

\subsection{Informal security argument}
We now present an informal security argument, with more formal demonstrations presented in Appendix~\ref{app:security}. The correctness follows from the correctness of the one-time pad and permutation (2-wise and uniform). For the privacy of the receiver's input against a dishonest sender we argue in the following way. Note that the sender receives $(\{P_t\}, Q)$ from trusted dealer during the key generation phase and  $\{x_t +r_t\}$  from the receiver in the computation phase. The security of $\{x_t\}$  reduces to the one-time pad, assuming the trusted dealer doesn't collude with the sender and generates uniformly random bits. The privacy of the sender's input against a dishonest receiver is the following. In a single run of protocol, the receiver obtains  $(\{r_t\}, \{s_{Q(t)}\})$ from the trusted dealer and  $\{v_{Q(t)} = P_t (u_t-y_t)\}$ from the sender. The sender's input privacy follows from the security guarantee of 2-wise random permutation  and under the same assumption of no-collusion between the trusted dealer and the receiver. The receiver checks if $v_{Q(t)} = s_{Q(t)}$, but it doesn't get any information about $t$ since $Q$ is a random permutation over $t$. 

It is not difficult to verify that the blind contact tracing protocol is perfectly secure against active adversaries and is optimal in terms of communication complexity ( as shown in~\cite{ishai2013power}. We also note that information theoretic security in two-party computation is not possible in the plain model unless certain assumptions such as trusted agents or shared secrets are made. Thus, our assumption of the trusted dealer can be justified in the sense that it allows us to achieve a strong security notion (information-theoretic) while preserving the efficiency of the protocol. Since the trusted dealer is never involved in the computation phase, as long as it do not collude with either of the parties, the protocol remain secure against any other attack by the trusted dealer.

\section{Defining the receiver for COVID-19}
Up until now, we have presented our protocol in terms of an arbitrary receiver and sender. However, in the implementation and depending on the objectives one needs to specify who the receiver is. Here we consider three cases specific to COVID-19.

\begin{enumerate}
    \item \textbf{An individual is the receiver and the government is the sender}. The privacy of the individual, in that case, is guaranteed straightforwardly by the security of blind contact tracing protocol. In practice, this requires direct communication between the individual and the government while the individuals would have to possess the database with their location trails.
    \item \textbf{Service Provider (IT/Telecom company) is the receiver, the governmental body is the sender}. This is a scenario under which the service provider is the receiver and it notifies the users whether they were in contact with an infected person or not. In practice, this is easier to implement than the previous scenario. Again, blind contact tracing protocol offers perfect security in that case under the assumption the service provider is not colluding with the government, which is in any case, the assumption whenever any such service is used.
    \item \textbf{Government is the receiver}. The scenario under which the government is the receiver might be particularly useful if the government wants to learn about the statistics of people who were in the proximity of infected patients. If it wishes so, the government could then either directly or via service provider inform the individual of their risk, without knowing their identity.
\end{enumerate}
\section{Implementation}
To determine the feasibility of our approach to applications, such us contract tracing for  COVID-19, we implement the blind contact tracing protocol using severless infrastructure available on cloud computing platforms.  The first crucial elements to consider in the implementation is the input data and the grid construction that we apply to it. The second one is the resource intensiveness. We discuss both in more detail in the following sections.

\subsection{Input data and grid construction}\label{grid}

The raw input to the protocol may come in a variety of forms such as GPS location information, GSM triangulation data, Bluetooth beacons, and router mac addresses amongst others. In many cases, we can expect the input to be a continuous value such as longitude and latitude. To make this data fit the paradigm of the blind contact tracing, we need to map locations to some discrete tiling over the surface of the earth. There are a few ways to do this such as constructing a square grid or hexagonal patterns, which are more efficient in terms of search queries, as seen in \cite{narayanan2011location}. Both approaches satisfy the protocols correctness. The second consideration is that the Earth is not flat. To avoid error we must scale the mapping of location to tiles depending on the location on earth we are considering. In \cite{vsedvenka2014privacy}, the Haversine approximation as outlined which approximates Earth as a sphere and has an approximation error rate of $< 0.1\%$. We used an alternative approach, recommended by the FCC \cite{officercode}, which uses an ellipsoid approximation to earth and has significantly less error over distance under a few hundred kilometers. To scale GPS coordinates to kilometers under this approximation we calculate two constants which we multiply the latitude and longitude by

\[
\begin{array}{ll}
    C_{\text{Lat}} &=  111.13209 - 0.56605\cos(2L\pi/180) + 0.0012\cos(4L\pi/180) \\
    C_{\text{Lon}} &= 111.41513\cos(L\pi/180) - 0.09455\cos(3L\pi/180) + 0.00012\cos(5L\pi/180)
\end{array}
\]

\noindent where $L$ is the longitude of the point of interest.

\subsection{Resource Intensiveness}
\label{sec:resource}

An important property of the proposed solution is its ability to scale to deal with contact tracing at a national level. To this end, it is trivial to show how the solution scales at a constant rate for the data. The approach is easily run in parallel across containers to meet the demand required. To get a ballpark estimate of the demand let us consider the scenario for a small, medium and large state based on estimated figures taken from Ireland, Poland and Japan roughly at the time of publication. We have also included similar comparisons at a synthetic city level of abstraction. We search a $3\times3$ overlapping grid of $7\times7$ meters each about a location to look for a match to avoid edge effects. We also check a time window of $40$ minutes around the timestamp associated with a location. Time is quantized into $20$ minute periods and we search for paths crossed in a single day for exemplary purposes. Of course, if we were to ignore edge effects we could reduce the over head by a small factor or if we wished to increase accuracy we could sample a finer resolution of accuracy about each point. This trade off is left to the users of the protocol.

%\begin{center}
%\begin{adjustbox}{width=\textwidth, rotate=0}
\begin{table}[h!]
\scalebox{0.92}{
\begin{tabular}{@{}cccccccccccccc@{}}
\toprule
&& \multicolumn{5}{l}{\textbf{Nation}} && \multicolumn{5}{l}{\textbf{City}}\\
\textbf{Size} &\phantom{a}& \textbf{Small} &\phantom{a}& \textbf{Medium} &\phantom{a}& \textbf{Large}&\phantom{a}& \textbf{Small} &\phantom{a}& \textbf{Medium} &\phantom{a}& \textbf{Large}\\ \midrule
 \textbf{Population (M)}&& 4.83 && 37.98 && 1339 && 1 && 5 && 20\\
 \textbf{Cases}&& 5364  && 4666  && 5172 && 100 && 200 && 500   \\
 \textbf{Keys (TB)}&& 6.1 $\times10^2$  && 4.2 $\times10^3$ && 1.1 $\times10^4$ &&2.35&& 23.5 && 235.7\\ 
 \textbf{Comms (TB)}&& 8.5 $\times10^2$  && 5.8 $\times10^3$ && 1.6 $\times10^4$ &&3.3&& 33.0 && 330.0\\ \bottomrule
\end{tabular}
}
\caption{ The table outlines the keys required (Keys) and communication overhead (Comms) to identify paths of the entire population of countries and cities that crossed with COVID-19 patients.}
\end{table}

%\end{adjustbox}
%\end{center}

\noindent While the above numbers are, of course, rough estimates it gives us confidence that the proposed solution is feasible leveraging available cloud services. These figures are also making the naive assumption that one would cross-check confirmed cases with the entire country, rather than for example cross-check on a state by state level. For example, on Amazon Web Services 25Gb/s communication interfaces are currently available and on Oracle Cloud up to 100Gb/s communication interfaces are available per instance.  Similar IO rates can be found with Microsoft Azure, Google Cloud, and other major cloud providers. In parallel systems we can transfer and compute many times this amount. While these requirements are certainly not negligible, they are technically feasible for deployment.

\section{Conclusions}
In this work we endeavour to overcome the potential challenges faced by computational security approaches in the context of contact tracing. The proposed protocol offers an information theoretic security and is optimal in terms of communication complexity. 
Furthermore, our solution has advantages in terms of effectiveness. It can use different sources of geolocation data such as GPS, cell-towers triangulation, Bluetooth or WiFi routers, and is not restricted to a single mode. Importantly it can be used on historic data collected by different means. It can be readily implemented on existing infrastructure by  governmental institutions and third parties that have access to geolocation data including telecommunication service providers or IT enterprises. As such, it does not require installation of a separate smartphone application. By using cloud computing infrastructure, we have developed a scalable key generation service to facilitate these efforts. Access to the service, along with a lightweight library in Python to leverage service, is available to the wider community upon request\footnote{Please email \url{covid19@oblivious.ai}}.

\section{Acknowledgements}

We would like to acknowledge Oracle Cloud, Amazon Web Services and Strategic-Blue for their support. All authors would like to thank Joe Fitzsimons for his insights and discussions. A.M gratefully acknowledges funding from the AFOSR MURI project ``Scalable Certification of Quantum Computing Devices and Networks''. R.P. thanks EPSRC (UK). A.M would also like to thank Ashima Arora for some helpful discussion regarding implementation. 

\bibliographystyle{unsrt}
\bibliography{biblio}

\appendix

\section{Security proofs}
\label{app:security}
An important question for public privacy is how difficult would it be for an adversarial government to query the database in order to perform mass surveillance and tracking of citizens.

We require a blind contact tracing protocol to satisfy two properties: correctness and input privacy. The former deals with the setting when all the parties (sender, receiver, dealer) follow the steps of the protocol as described in Protocol~\ref{main protocol}. For the latter, we consider two sub-cases - a) the receiver's input privacy when the sender is cheating and b) the sender's input privacy when the receiver is cheating. It is important to note that we assume that a trusted dealer behaves honestly and does not collaborate with any malicious party. 

\begin{definition}[Correctness]
We say that a blind contact tracing protocol $\mathcal{P}_{x,y}$ (where both the parties sender and receiver are honest) is correct if for all inputs $(x,y)$ the parties output from running the protocol $\mathcal{P}_{x,y}$ is exactly equal to the output of parties in ideal blind contact tracing functionality $\mathcal{S}_{bct}(x, y)$.
\end{definition}

\begin{theorem}[Correctness]
The blind contact tracing protocol shown in Protocol~\ref{main protocol} is perfectly correct. That is for every random 2-wise permutation ${P_t: X \mapsto X}$, uniformly random permutation $Q \in \mathcal{S}_n$, every message ${r_t \in X}$, $s_{Q(t)}$ (with $s_{Q(t)} = P_t(r(t))$) and every input $\{x_t \in X | t \in [n]\}$, $\{y_t \in X | t \in [n]\}$ it holds that receiver obtains $N  := \sum_{Q(t)} w_{Q(t)} = \sum_{k} z_{k}$, where $z_{k}$ is $1$ if $x_{k} = y_{k}$ otherwise  $0$. On the other hand, the sender obtains no output. 
\end{theorem}

\begin{proof}
The correctness follows straightforwardly from the property of 2-wise permutation ${P_t}$ and uniformly random permutation $Q$ as well as the correctness of one-time pad. The output $N$ is incremented on the receiver side if and only if $x_t - y_{t}$ is nonzero. To see this, let's consider
\begin{equation}
\label{eq:correctness1}
    v_{Q(t)} - s_{Q(t)}  = P_t(u_t - y_t) - P_t(r_t) = P_t(u_t - y_t -r_t)
\end{equation}
As we know, the receiver computes $u_t = x_t + r_t$, therefore Eq.~\ref{eq:correctness1} can be equivalently written as
\begin{equation}
\label{eq:correctness2}
    v_{Q(t)} - s_{Q(t)}  = P_t(x_t  - y_t)
\end{equation}
Therefore, $v_{Q(t)} = s_{Q(t)}$ if and only if $x_t  = y_t$. Thus, $N := \sum_{Q(t)} w_{Q(t)} = \sum_{Q(t)} z_{Q(t)}$. 
\end{proof}

\begin{definition}[Privacy against malicious party]
We say that a blind contact tracing protocol $\mathcal{P}(x,y)$ is: 
\begin{enumerate}
    \item secure against dishonest receiver if it does not leak anything about sender's input $y$ except $N$ (total number of $i$'s where $x_i = y_i$).
    \item secure against dishonest sender if it does not leak anything about receiver's input $x$.  
\end{enumerate}
\end{definition}

\begin{theorem}[Privacy against malicious party]
The blind contact tracing protocol shown in Protocol~\ref{main protocol} is perfectly secure against a malicious sender or a malicious receiver (assuming the trusted dealer doesn't collude with any adversarial party). 
\end{theorem}

\begin{proof}

The security proof of Protocol~\ref{main protocol} is only a slight modification of the one presented in \cite{ishai2013power} and uses a simulation proof technique \cite{lindell2017simulate}. We shall construct a simulator that creates a view for the adversary's which is indistinguishable from the real execution of the protocol.
The simulator holds all the information from the key generation part i.e. $(\{P_t\},Q, \{r_t\}, \{s_{Q(t)}\})$. Now consider the two cases of malicious sender or malicious receiver.
\begin{enumerate}
\item Malicious sender:  The simulator generates randomly $u_t$ and send it to the adversary (the sender). The simulator receives $v_{Q(t)}$ from the adversary and computes $y_t = u_t - P^{-1}_t (v_{Q(t)})$. It sends $y_t$ to the ideal blind contact tracing functionality $\mathcal{S}_{bct}$.
\item Malicious receiver - The simulator calculates $x_t= u_t - r_t$.  and inputs it to the ideal blind contact tracing functionality $\mathcal{S}_{bct}$. If the ideal functionality outputs $1$, the simulator sends $s_{Q(t)}$ to the adversary. If it outputs 0, the simulator chooses randomly $v_{Q(t)}$ such that $v_{Q(t)} \neq  s_{Q(t)}$ and sends it to the adversary.
\end{enumerate}
\end{proof}

\end{document}